\RequirePackage{fix-cm}
\documentclass[smallextended]{svjour3}       
\smartqed  

\usepackage[round, sort]{natbib}
\usepackage{graphicx}
\usepackage{amsmath,amssymb}
\usepackage{amsthm}
\usepackage{arydshln}
\usepackage{rotate,rotating}
\usepackage{epsf,epic,eepic}
\usepackage{bmpsize}
\usepackage{multirow} 
\usepackage{mathrsfs}

\usepackage{color}

\newcommand{\mc}[1]{\mathcal{#1}}
\journalname{Journal of Mathematical Biology}
\begin{document}

\title{Global stability properties of renewal epidemic models}

\author{Michael T. Meehan        \and
				Daniel G. Cocks \and
				Johannes M{\"u}ller \and
				Emma S. McBryde
}


\institute{M. Meehan, E. McBryde \at
              Australian Institute of Tropical Health and Medicine, James Cook University \\
              \email{michael.meehan1@jcu.edu.au}           
           \and
           D. Cocks \at
              Research School of Science and Engineering, Australian National University 
						\and
						J. M{\"u}ller \at
						Centre for Mathematical Sciences, Technische Universit{\"a}t M{\"u}nchen, and 
						Institute of Computational Biology, German Research Center for Environmental Health, M{\"u}nchen
}



\date{Received: date / Accepted: date}

\maketitle

\begin{abstract}
We investigate the global dynamics of a general Kermack-McKendrick-type epidemic model formulated in terms of a system of renewal equations. Specifically, we consider a renewal model for which both the force of infection and the infected removal rates are arbitrary functions of the infection age, $\tau$, and use the direct Lyapunov method to establish the global asymptotic stability of the equilibrium solutions. In particular, we show that the basic reproduction number, $R_0$, represents a sharp threshold parameter such that for $R_0\leq 1$, the infection-free equilibrium is globally asymptotically stable; whereas the endemic equilibrium becomes globally asymptotically stable when $R_0 > 1$, i.e. when it exists. 


\end{abstract}




\section{Introduction}
\label{sec:Introduction}


The classic Kermack-McKendrick paper~\citep{kermack1927contribution} is a seminal contribution to the mathematical theory of epidemic modelling. Within, the authors formulate a general epidemic model in which the infectiousness of infected individuals and the rate at which they recover or are removed is an arbitrary function of the infection age, $\tau$; from this, they derive several fundamental results including the conditions for an epidemic outbreak and the final size equation. As a consequence of their general formulation, the analysis and conclusions of the Kermack-McKendrick paper encompass a wide class of epidemic models, including countless incarnations that have since appeared in the infectious diseases modelling literature (e.g. the SIR and SEIR models).

In this article we revisit the classic Kermack-McKendrick model~\citep{kermack1927contribution} and further investigate the system properties and global dynamics in the presence of demographic influences. Our main result, which is derived in section~\ref{sec:stability}, is to show that the basic reproduction number $R_0$ represents a sharp threshold parameter that determines the global stability of the infection-free and endemic equilibria. Specifically, we find that when $R_0 \leq 1$ the infection-free equilibrium point is the unique equilibrium in the nonnegative orthant and is globally asymptotically stable within this region. Conversely, when $R_0 > 1$ an endemic solution emerges in the interior of this region which is globally asymptotically stable away from the invariant $S$-axis. Both of these results are proved by the direct Lyapunov method, that is, by identifying appropriate Lyapunov functionals.

Lyapunov functions have previously been used to establish the global asymptotic stability properties of SIR, SIS and SIRS models (see e.g.~\citep{Korobeinikov2004}) for which the population is either constant~\citep{Korobeinikov2002955,ORegan2010446} or varying~\citep{li1995global,Li1999191}. These results have also been extended to SEIR and SEIS models in~\citep{li1995global,Li1999191,fan2001global,McCluskey2008518}, and epidemic models with multiple parallel infectious stages~\citep{Korobeinikov2008} or strains~\citep{Bichara2013}. However, these results have each been established within the context of compartmental-type epidemic models for which the per-capita flow rates between the stages of infection are assumed to be constant and infectiousness is fixed for the duration of their infectious period.

Only recently, by using an approach that relied on both the direct Lyapunov method and semigroup theory, were~\citep{magal2010lyapunov} able to determine the global stability properties of equilibria in infection-age models. This work has since been expanded~\citep{McCluskey2008518,McCluskey201055,McCluskey2009603} and extended to models with general incidence functions~\citep{Huang2011,MCCLUSKEY20103106,CHEN201616,SOUFIANE20161211} and multiple parallel infectious stages~\citep{wang2012global} or strains~\citep{MARTCHEVA2013225}. Here, we provide an alternative treatment given in terms of the original renewal formulation of the Kermack-McKendrick model.


In the next section we briefly describe the renewal system variables, parameters and their governing equations, and then discuss the system phase-space. Then, in section~\ref{sec:stability}, we derive the main result of this article where we introduce a set of Lyapunov functionals which we use to establish the global asymptotic stability of the infection-free and endemic equilibria.

\section{Model description}
\label{sec:model}
In the renewal formulation of the Kermack-McKendrick model we need only explicitly consider a class of susceptible (i.e. infection-na{\"i}ve) individuals, S, who each experience a time-dependent force of infection $F(t)$.\footnote{The dynamics of the class of infected individuals, I, is implicitly captured through the force of infection, $F(t)$. 
} By definition, the force of infection is the per-capita rate at which susceptibles become infected. Therefore the incidence at time $t$, $v(t)$, is given by
\begin{equation}
v(t) = F(t) S(t),\nonumber
\end{equation}
where $S(t)$ is the number of susceptibles and $v(t)$ describes the rate at which new infected individuals appear at time $t$. Assuming then that individuals who have been infected for $\tau$ units of time on average contribute an amount $A(\tau)$ to the force of infection, we find that the total force of infection at time $t$, $F(t)$, can be written in terms of a renewal equation:
\begin{align}
F(t) &= \int_0^\infty A(\tau) v(t-\tau)\,d\tau,\nonumber\\
&= \int_0^\infty A(\tau) F(t-\tau)S(t-\tau)\,d\tau.\nonumber
\end{align}
Here $v(t-\tau)$ represents the number of individuals who became infected at time $t-\tau$. 

In general, the infectivity kernel $A\geq 0$ is an arbitrary function of the infection age $\tau$ 
whose definition motivates us, in analogy with~\citep{magal2010lyapunov}, to define
\begin{equation}
\bar{\tau} = \mathrm{sup}\left\{\tau\geq 0\: :\: A(\tau) > 0\right\},\label{eq:taubar}
\end{equation}
the maximum infection age at which an individual can contribute to the force of infection. In this case we need only look back to this maximum infection age to calculate $F(t)$:
\begin{align}
F(t) &= \int_0^{\bar\tau}A(\tau) v(t-\tau)\,d\tau,\nonumber\\
&= \int_0^{\bar\tau} A(\tau) F(t-\tau)S(t-\tau)\,d\tau.\label{eq:F}
\end{align}

To complete the model description we assume that in addition to removal by infection, individuals are recruited into the susceptible class at a constant rate $\lambda$ and die naturally at a constant per-capita rate $\mu$. Combining these rates, we find that 
\begin{equation}
\frac{dS(t)}{dt} = \lambda - \mu S(t) - F(t)S(t),\label{eq:dS}
\end{equation}
assuming that infection leads to permanent immunity. Given~\eqref{eq:dS}, it is straightforward to show that $S(t) > 0$ for $t > 0$ provided it is nonnegative initially.

An important parameter that governs the system trajectory is the basic reproduction number $R_0$, defined as the expected number of secondary cases caused by a single (typical) infectious individual in a fully susceptible population. Given the definition of $A(\tau)$ and the expression for $F(t)$ (equation~\eqref{eq:F}), the functional form of $R_0$ is naturally given by
\begin{equation}
R_0 = S^0\int_0^{\bar\tau} A(\tau)\,d\tau\label{eq:R0}
\end{equation}
where $S^0=\lambda/\mu$ is the steady-state susceptible population in the absence of infection (see below).

Before introducing suitable initial conditions for the system, we emphasize that in order to solve~\eqref{eq:F} and~\eqref{eq:dS} we must have knowledge of the entire past history of $F$ and $S$ over the interval $\tau \in [-\bar\tau, 0]$. Therefore, the state of our system $P = (\mc S,\mc F)$ belongs to an infinite-dimensional phase-space $\Omega$, which can appropriately be chosen as
\begin{equation}
\Omega = C^0_+([-\bar\tau, 0]) \times L^1_+(-\bar\tau, 0).\nonumber
\end{equation}
With this choice of state-space, standard arguments show that the model~\eqref{eq:F}-\eqref{eq:dS} is well defined. In particular, for $(\mc S, \mc F)\in\Omega$ we can construct an extension $(S, F)\in C^0_+([-\bar\tau, \infty)) \times L^1_{loc}(-\bar\tau, \infty)$ that solves the model equations. 

Given $\Omega$ then, a suitable choice of initial conditions is given by
\begin{equation}
\mc S_0 \in C^0_+([-\bar\tau,0]) \qquad \mbox{and} \qquad \mc F_0\in L^1_+(-\bar\tau,0).\nonumber
\end{equation}
By assuming that the initial histories $\mc S_0$ and $\mc F_0$ are bounded, combined with moderate restrictions on the infectivity kernel $A$, we guarantee that extensions of the initial conditions generated by the model equations~\eqref{eq:F} and~\eqref{eq:dS} remain bounded and continuous. 

Next, the trajectory is given by  $(\mc S_t(\cdot), \mc F_t(\cdot))\in\Omega$ where
\begin{equation}
\mc S_t(s) = S(t+s),\quad \mc F_t(s) = F(t+s),\qquad s\in[-\bar\tau, 0].\nonumber
\end{equation}
Using~\eqref{eq:F} and~\eqref{eq:dS} we can show that $(\mc S_t, \mc F_t)\in C^1([-\bar\tau,0])\times W^{1,1}(-\bar\tau,0)$ for $t > \bar\tau$, such that we have compactness for $t > \bar\tau$. Moreover, since the trajectory is bounded, the $\omega$-limit set of the system~\eqref{eq:F}-\eqref{eq:dS} is non-empty.

We point out that in this notation, the pair $(\mc S_t(0), \mc F_t(0)) = (S(t), F(t))$ represents the most recent value in the history of a state along the system trajectory at time $t$, namely $(\mc S_t(\cdot), \mc F_t(\cdot))\in\Omega$. In this case the model equations~\eqref{eq:F}-\eqref{eq:dS} can be understood as rules for updating the most recent values of the histories of $F$ and $S$ respectively. 



Of particular interest within the larger, forward invariant phase-space $\Omega$, is the interior region $\widehat{\Omega}\subset\Omega$ for which new infections will arise either at the present time or at some time in the future. That is,
\begin{equation}
\widehat{\Omega} = \left\{(\mc S, \mc F) \in \Omega \: : \: \exists \, a\in [0,\bar\tau] \: \:\mbox{s.t.} \: \int_{0}^{\bar\tau} A(\tau + a) \mc F(-\tau) \mc S(-\tau)\,ds > 0 \right\}.\nonumber
\end{equation}

Conversely, we can also segregate the boundary, $\partial\Omega$, of the phase-space, for which no new infections can arise and for which the infection will be eradicated 
\begin{equation}
\partial\Omega = \Omega \setminus \widehat{\Omega}.\nonumber
\end{equation}

Finally, it is easy to verify that the fixed states of the system~\eqref{eq:F}-\eqref{eq:dS} are given by
\begin{align}
P_0 &= (\mc S^0, \mc F^0) = (S^0, F^0) = \left(\frac{\lambda}{\mu},0\right), \nonumber\\
\bar{P} &= (\bar{\mc S},\bar{\mc F}) = (\bar S, \bar F) = \left(\frac{\lambda}{\mu R_0},\mu(R_0 - 1)\right).
\end{align}
Importantly, we see that the endemic equilibrium point, $\bar{P}$, only exists in the interior region $\widehat{\Omega}$ for $R_0 >1$; for the limiting case $R_0 = 1$, the endemic and infection-free equilibria coincide.


Ultimately, our goal will be to establish that i) when $R_0 \leq 1$ all system trajectories of~\eqref{eq:F}-\eqref{eq:dS} within $\Omega$ asymptotically approach the infection-free equilibrium point $P_0 \in \partial\Omega$ and ii) when $R_0 > 1$ trajectories that originate in $\Omega$ asymptotically approach the endemic equilibrium $\bar{P} \in \widehat{\Omega}$, except those that originate in $\partial\Omega$ which approach $P_0$.

\section{Global stability analysis}
\label{sec:stability}


\subsection{Infection-free equilibrium}
\begin{theorem}
\label{the:ife}
Let $A\in C^1$. The infection-free equilibrium point $P_0$ of the system~\eqref{eq:F}-\eqref{eq:dS} is globally asymptotically stable in $\Omega$ for $R_0 \leq 1$. However, if $R_0 > 1$, solutions of~\eqref{eq:F}-\eqref{eq:dS} starting sufficiently close to $P_0$ in $\Omega$ move away from $P_0$, except those starting within the boundary region $\partial \Omega$ which approach $P_0$.
\end{theorem}

\begin{proof}[Proof of Theorem~\ref{the:ife}]
To verify theorem~\ref{the:ife} we first point out that the model equations~\eqref{eq:F}-\eqref{eq:dS} induce a continuous semiflow $\Phi_t  :  \Omega \rightarrow \Omega$. Hence, 
if we define $D = \Phi_{\bar\tau}(\Omega)$, from~\eqref{eq:dS} we have $\mc S(0) > 0$ for $(\mc S,\mc F)\in D$. Note, $D$ is closed and forward invariant, and any trajectory originating in $\Omega$ enters $D$ either at, or before $t = \bar\tau$. 

Consider the Lyapunov functional $U :  D\rightarrow \mathbb{R}_+$ defined by
\begin{equation}
U(\mc S,\mc F) = g\left(\frac{\mc S(0)}{S^0}\right) + \int_0^{\bar\tau} \eta(\tau)\mc F(-\tau)\mc S(-\tau)\,d\tau\nonumber
\end{equation}
where
\begin{equation}
g(x) = x - 1 - \log x \qquad \mbox{and} \qquad \eta(\tau) = \int_\tau^{\bar\tau } A(s)\,ds.\label{eq:getadef}
\end{equation}
In particular we have $\eta(\bar\tau)= 0$, 
\begin{equation}
\eta(0) = \frac{R_0}{S^0} \qquad \mbox{and} \qquad \eta'(\tau) = - A(\tau) \label{eq:etacons}
\end{equation}
where a $'$ denotes differentiation with respect to $\tau$. Importantly, the functional $U(\mc S,\mc F) \geq 0$ is well defined since $\mc S(0) > 0$, and has a global minimum at the infection-free equilibrium $P_0$.

Next, let $(\mc S_t(\cdot),\mc F_t(\cdot))$ be a trajectory of the model~\eqref{eq:F}-\eqref{eq:dS} with initial condition in $D$. With $\mc S_t(s) = S(t + s)$ and $\mc F_t(s) = F(t+s)$ we may write
\begin{align}
U(\mc S_t(\cdot),\mc F_t(\cdot)) &= g\left(\frac{\mc S_t(0)}{S^0}\right) + \int_0^{\bar\tau} \eta(\tau) \mc F_t(-\tau) \mc S_t(-\tau)\,d\tau,\nonumber\\
&=g\left(\frac{S(t)}{S^0}\right) + \int_0^{\bar\tau} \eta(\tau) F(t-\tau)S(t-\tau)\,d\tau.\nonumber
\end{align}
In order to compute the time derivative of $U(\mc S_t, \mc F_t)$ we rewrite this as
\begin{equation}
U(\mc S_t(\cdot), \mc F_t(\cdot)) = g\left(\frac{S(t)}{S^0}\right) + \int_{t-\bar\tau}^{t} \eta(t-s) F(s)S(s)\,ds.\label{eq:U}
\end{equation}



Differentiating each term in~\eqref{eq:U} along system trajectories separately, we first have
\begin{align}
\frac{d}{dt}\left[g\left(\frac{S(t)}{S^0}\right)\right] &= \left(\frac{1}{S^0} - \frac{1}{S(t)}\right)\frac{dS(t)}{dt},\nonumber\\
&=  \frac{\lambda}{S^0} - \mu\,\frac{S(t)}{S^0} - F(t)\,\frac{S(t)}{S^0} - \frac{\lambda}{S(t)} + \mu  + F(t),\nonumber\\
&= \mu \left(2 - \frac{S(t)}{S^0} - \frac{S^0}{S(t)}\right) - F(t)\,\frac{S(t)}{S^0} + F(t),\nonumber\\
&= -\mu\,\frac{S(t)}{S^0}\left(1 - \frac{S^0}{S(t)}\right)^2 - F(t)\left(\frac{S(t)}{S^0} - 1\right)\label{eq:U1dot}
\end{align}
where in the second line we have substituted in the identity $\lambda = \mu S^0$.


Next we differentiate the second term in~\eqref{eq:U} to obtain
\begin{align}
\frac{d}{dt}\left(\int_{t-\bar\tau}^t \eta(t-s)F(s)S(s)\,ds\right) &= \eta(0)F(t)S(t) - \eta(\bar\tau)F(t-\tau)S(t-\tau) \nonumber\\
&\qquad + \int_{t-\bar\tau}^t \frac{d\eta(t-s)}{dt}F(s)S(s)\,ds,\nonumber\\
&=R_0 F(t)\,\frac{S(t)}{S^0} - \int_{t-\bar\tau}^t A(t-s)F(s)S(s)\,ds,\nonumber\\
&= R_0 F(t)\,\frac{S(t)}{S^0}  - F(t)\label{eq:U2dot}
\end{align}
where in the second line we have substituted in~\eqref{eq:etacons} and in the last line we have used the definition of $F(t)$, equation~\eqref{eq:F}.

Finally, combining~\eqref{eq:U1dot} and~\eqref{eq:U2dot} yields
\begin{equation}
\frac{d}{dt}U(\mc S_t, \mc F_t) = -\mu\,\frac{\mc S_t(0)}{S^0}\left(1 - \frac{S^0}{\mc S_t(0)}\right)^2 -\left(1 - R_0\right)\mc F_t(0)\,\frac{\mc S_t(0)}{S^0}.\label{eq:Udotfinal}
\end{equation}
We emphasize that we know for a trajectory $(\mc S_t,\mc F_t)\in D \subset \Omega$, that for $t > 0$ we already have $\mc F_t\in C^0([-\bar\tau,0])$, such that this expression is well defined and $U$ is a proper Lyapunov functional on the closed domain $D$.

Importantly, for $R_0 \leq 1$ we have $dU/dt \leq 0$. The derivative $\dot{U}(t) = 0$ if and only if $\mc S_t(0) = S^0$ and either (a) $R_0 = 1$ or (b) $\mc{F}_t(0) = 0$. Therefore, the largest invariant subset in $\Omega$ for which $\dot{U} = 0$ is the singleton $\left\{P_0\right\}$. 
As $A(\tau)$ is smooth, the orbit is eventually precompact. Hence, by the infinite-dimensional form of LaSalle's extension of Lyapunov's global asymptotic stability theorem~\citep[Theorem~5.17]{hsmith_textbook}, the infection-free equilibrium point $P_0$ is globally asymptotically stable in $\Omega$ for $R_0 \leq 1$.

Conversely, if $R_0 > 1$ and $\mc F_t(0) > 0$, the derivative $\dot{U} > 0$ if $S(t)$ is sufficiently close to $S^0$. 
Therefore, solutions starting sufficiently close to the infection-free equilibrium point $P_0$ leave a neighbourhood of $P_0$, except those starting in $\partial\Omega$. Since $\dot{U} \leq 0$ for solutions starting in $\partial\Omega$ these solutions approach $P_0$ through this subspace.

\end{proof}

\subsection{Endemic equilibrium}

\begin{theorem}
\label{the:endemic} Let $A\in C^1$. If $R_0 > 1$ the endemic equilibrium point $\bar{P}$ is globally asymptotically stable in $\widehat{\Omega}$ (i.e. away from the boundary region $\partial \Omega$).
\end{theorem}

\begin{proof}[Proof of Theorem~\ref{the:endemic}]
First, in theorem~\ref{the:ife} we observed that $F(t)$ for $t > 0$ is bounded away from zero when $R_0 > 1$, such that that for $R_0 > 1$ the semiflow $\Phi_t : \widehat\Omega \rightarrow \widehat\Omega$. Therefore, in analogy with theorem~\ref{the:ife} we define $\widehat{D} = \Phi_{\bar\tau}(\widehat\Omega)$ which is a closed, forward-invariant set for $R_0 > 1$. Moreover $\mc S, \mc F > 0$ for $(\mc S,\mc F)\in\widehat{D}$ and any trajectory originating in $\widehat\Omega$ enters $\widehat{D}$ at the latest at time $t = \bar\tau$, provided $R_0 > 1$.

In this case we define $W\, : \,  \widehat{D}\rightarrow \mathbb{R}_+$
\begin{equation}
W(\mc S,\mc F) = g\left(\frac{\mc S(0)}{\bar{S}}\right) + \int_0^{\bar\tau }{\chi(\tau)g\left(\frac{\mc F(-\tau)\mc S(-\tau)}{\bar{F}\bar{S}}\right)}\,d\tau\nonumber
\end{equation}
where $g(x)$ has been defined previously in~\eqref{eq:getadef} and
\begin{equation}
\chi(\tau) = \bar{F}\bar{S}\int_\tau^{\bar\tau } A(s)\,ds.\nonumber
\end{equation}
Immediately we have that $\chi(\bar\tau) = 0$,
\begin{equation}
\chi(0) = \bar{F} \qquad \mbox{and} \qquad \chi'(\tau) = -\bar{F}\bar{S}A(\tau).\nonumber
\end{equation}
Once again, note that $W(\mc S,\mc F)$ is well defined on $\widehat{D}$.

Similar to before, we let $(\mc S_t(\cdot),\mc F_t(\cdot))$ be a trajectory of the model with initial condition in $\widehat{D}$ and adopt the notation $\mc S_t(s) = S(t+s)$ and $\mc F_t(s) = F(t+s)$. We may then write
\begin{equation}
W(\mc S_t(\cdot),\mc F_t(\cdot)) = g\left(\frac{S(t)}{\bar{S}}\right) + \int_0^{\bar\tau} \chi(\tau)\,g\left(\frac{F(t-\tau)S(t-\tau)}{\bar{F}\bar{S}}\right)\,d\tau\nonumber
\end{equation}
which we at once rewrite as
\begin{equation}
W(\mc S_t(\cdot),\mc F_t(\cdot)) = g\left(\frac{S(t)}{\bar{S}}\right) + \int_{t-\bar\tau}^{t} \chi(t-s)\,g\left(\frac{F(s)S(s)}{\bar{F}\bar{S}}\right)\,d\tau.\label{eq:W}
\end{equation}


Once again we differentiate each term separately. Beginning with the first term in~\eqref{eq:W} we have
\begin{align}
\frac{d}{dt}\left[g\left(\frac{S(t)}{\bar{S}}\right)\right] &=\left(\frac{1}{\bar{S}} - \frac{1}{S(t)}\right)\frac{dS(t)}{dt}, \nonumber\\
&= \frac{\lambda}{\bar{S}} - \mu\,\frac{S(t)}{\bar{S}} - F(t)\,\frac{S(t)}{\bar{S}} - \frac{\lambda}{S(t)} + \mu + F(t),\nonumber\\
& = \mu\left(2 - \frac{S(t)}{\bar{S}} - \frac{\bar{S}}{S(t)}\right) + \bar{F}\left(1 - \frac{\bar{S}}{S(t)}\right) + F(t)\left(1 - \frac{S(t)}{\bar{S}}\right),\nonumber\\
&= -\mu\,\frac{S(t)}{\bar{S}}\left(1 - \frac{\bar{S}}{S(t)}\right)^2 + \bar{F}\left(1 - \frac{\bar{S}}{S(t)}\right) + F(t)\left(1 - \frac{S(t)}{\bar{S}}\right) \label{eq:W1dot}
\end{align}
where in the third line we have substituted in the identity $\lambda = \mu\bar{S} + \bar{F}\bar{S}$. 

Turning to the second term we find
\begin{align}
\frac{d}{dt}\left[\int_{t-\bar\tau} ^{t} \chi(t-s)\,g\left(\frac{F(s)S(s)}{\bar{F}\bar{S}}\right)\,d\tau\right] &= \chi(0)\,g\left(\frac{F(t)S(t)}{\bar{F}\bar{S}}\right) - 
\chi(\bar\tau)\,g\left(\frac{F(t-\bar\tau)S(t-\bar\tau)}{\bar{F}\bar{S}}\right)\nonumber\\
&\qquad + \int_{t-\bar\tau}^t \frac{d\chi(t-s)}{dt} g\left(\frac{F(s)S(s)}{\bar{F}\bar{S}}\right)\,ds,\nonumber\\
&= \bar{F}\,g\left(\frac{F(t)S(t)}{\bar{F}\bar{S}}\right) - \bar{F}\bar{S} \int_{t-\bar\tau}^t A(t-s) g\left(\frac{F(s)S(s)}{\bar{F}\bar{S}}\right)\,ds.\nonumber
\end{align}
Substituting in the definition $g(x) = x - 1 - \log x$ and equation~\eqref{eq:F} this expression becomes
\begin{align}
&\frac{d}{dt}\left[\int_{t-\bar\tau} ^{t} \chi(t-s)\,g\left(\frac{F(s)S(s)}{\bar{F}\bar{S}}\right)\,d\tau\right] \nonumber\\
&\qquad\qquad = F(t)\left(\frac{S(t)}{\bar{S}} - 1\right) - \bar{F}\left[\log\left(\frac{F(t)S(t)}{\bar{F}\bar{S}}\right) - \bar{S}\int_{t-\bar\tau}^t A(t-s)\log\left(\frac{F(s)S(s)}{\bar{F}\bar{S}}\right)\,ds   \right].
\end{align}

The final term in the square brackets can be bounded using Jensen's inequality\footnote{For a concave function $\varphi(\cdot)$ the following inequality holds~\citep{Jensen1906}:
\begin{equation}
\varphi\left(\int_0^\infty h(t)f(t) \,dt\right) \geq \int_0^\infty h(t)\varphi\left(f(t)\right)\,dt\nonumber
\end{equation}
where $h(t)$ is a normalized probability distribution.}:
\begin{align}
\bar{S}\int_{t-\bar\tau}^t A(t-s)\log\left(\frac{F(s)S(s)}{\bar{F}\bar{S}}\right)\,ds &\leq \log\left[\bar{S}\int_{t-\bar\tau}^t A(t-s)\frac{F(s)S(s)}{\bar{F}\bar{S}}\,ds\right],\nonumber\\
&= \log\left(\frac{F(t)}{\bar{F}}\right).\nonumber
\end{align}
 Importantly, we note that equality between the left- and right-hand sides occurs if and only if $F(t)S(t) = \bar{F}\bar{S}$. Substituting this result back into the expression above we find
\begin{align}
&\log\left(\frac{F(t)S(t)}{\bar{F}\bar{S}}\right) - \bar{S}\int_{t-\bar\tau}^t A(t-s)\log\left(\frac{F(s)S(s)}{\bar{F}\bar{S}}\right)\,ds\nonumber\\
&\quad \geq \log\left(\frac{F(t)S(t)}{\bar{F}\bar{S}}\right) - \log\left(\frac{F(t)}{\bar{F}}\right),\nonumber\\
&\quad = \log\left(\frac{S(t)}{\bar{S}}\right),\nonumber\\
&\quad \geq 1 - \frac{\bar{S}}{S(t)}\nonumber
\end{align}
where in the last line we have used $\log x \geq 1 - \frac{1}{x}$, where equality requires $S(t) = \bar{S}$. 

This condition implies that
\begin{equation}
\frac{d}{dt}\left[\int_{t-\bar\tau} ^{t} \chi(t-s)\,g\left(\frac{F(s)S(s)}{\bar{F}\bar{S}}\right)\,d\tau\right] \leq F(t)\left(\frac{S(t)}{\bar{S}	} - 1\right) - \bar{F}\left(1 - \frac{\bar{S}}{S(t)}\right).\label{eq:W3dotbound}
\end{equation}

Finally, combining~\eqref{eq:W1dot} and~\eqref{eq:W3dotbound} yields
\begin{align}
\frac{d}{dt}W(\mc S_t,\mc F_t) &\leq -\mu\,\frac{\mc S_t(0)}{\bar{S}}\left(1 - \frac{\bar{S}}{\mc S_t(0)}\right)^2,\nonumber\\
&\leq 0.\label{eq:Wdotfinal}
\end{align}

From equation~\eqref{eq:Wdotfinal} we see that the largest invariant subset in $\widehat{\Omega}$ for which $\dot{W} = 0$ consists only of the endemic equilibrium point $\bar{P}$. 
As $A(\tau)$ is smooth, the orbit is eventually precompact.
Hence, by LaSalle's extension of Lyapunov's asymptotic stability theorem, the endemic equilibrium point $\bar{P}$ is globally asymptotically stable in $\widehat{\Omega}$.


\end{proof}

 \bibliographystyle{spbasic} 
\bibliography{References}





\end{document}